\documentclass{article}
\usepackage{amssymb}
\usepackage{amsfonts}
\usepackage{amsmath,amsthm}
\theoremstyle{plain}
\newtheorem*{theorem*}{Theorem}
\newtheorem{prop}{Proposition}
\theoremstyle{definition}
\newtheorem{definition}{Definition}
\theoremstyle{remark}
\newtheorem*{remark}{Remark}
\usepackage{amsmath,bm}
\usepackage{cite}
\usepackage{upgreek}
\newcommand{\testleft}{\leftarrow\!\shortmid}
\newcommand{\tend}[1]{\hbox{\oalign{${#1}$\crcr\hidewidth$\scriptscriptstyle\bm{\sim}$\hidewidth}}}
\makeatletter
\providecommand{\LyX}{L\kern-.1667em\lower.25em\hbox{Y}\kern-.125emX\@}
\newfont{\Bb}{msbm10}
\newfont{\Sc}{eusm10}
\newfont{\Fr}{eufm10}
\makeatother 
\DeclareMathOperator{\End}{End}

\newcommand{\ed}{\end{document}}
\newcommand{\isom}{\cong}

\begin{document}

\title{Clifford Algebraic approach to the De Donder-Weyl Hamiltonian theory}

\author{ M.C.B. Fernandes \\ Instituto de F\'{\i}sica, Universidade de Bras\'{\i}lia 70910-900,\\
N\'{u}cleo de Relatividade e Teoria de Part\'{i}culas \\
 Bras\'{\i}lia, DF, Brazil \\ (e-mail: mcezar@unb.br)}

\maketitle

\begin{abstract}
The Clifford algebraic formulation of the Duffin-Kemmer-Petiau (DKP) algebras is applied to recast the De Donder-Weyl Hamiltonian (DWH) theory  as an algebraic
description independent of the  matrix representation of the DKP algebra. We show that the DWH equations for antisymmetric fields arise out of the action of the DKP
algebra on certain invariant subspaces of the Clifford algebra which carry the representations of the fields. The matrix representation-free formula for the bracket associated with the DKP form of the DWH equations is also derived. This bracket satisfies a generalization of the standard properties of the Poisson bracket.
\end{abstract}

\maketitle


\section{Introduction}

In the year 2000, a matrix formulation of the De Donder-Weyl Hamiltonian field equations using the matrix form of the $\beta$ generators of the  DKP algebra  was established \cite{Igor1}. It was then observed that the $\beta$-matrix generators take part in the DWH field equations as the analogues of the symplectic matrix of particle mechanics. This result has linked the DWH theory \cite{Cara, Weyl, DeDonder} to the DKP matrix algebra \cite{Petiau,Duffin,Kemmer1,Kemmer2,Corson} thus pointing to a new research front related to the algebraic structure underlying the DWH field equations.

The DWH theory has been an active field of research in mathematical physics with a wide range of results.  The distinguished appeal of this theory as a theory that consists of a Hamiltonian system of \emph{covariant}  equations makes it especially attractive from the point of view of physics \cite{Vey, 1912.13363,1911.05962,0208036}, notably in the contexts of quantization of general relativity \cite{gr-qc/9810076,gr-qc/0004066,gr-qc/0012074}, quantization of field theory in curved space-time \cite{11Symmetry}, quantum Yang-Mills theory \cite{arXiv:1805.o5279,1706.01766} and quantum gravity \cite{Igor7}.

The development of the DWH theory has a long history that we will not review in this article. For a list of references, the reader could consult
\cite{Tulc,Sny,Gaw,Kij,KijW,KijT,Corn,Gia,Igor2,Gotay,Kastrup}. Recently, the DKP form of the DWH equations for multicomponent fields was applied in  \cite{Pietrzyk} to write Einstein's equations in the study of the polysymplectic integrator in numerical general relativity.

The aim of the present work is to contribute to the understanding of the relationship between the DKP algebra and the DWH equations.

The Clifford algebra will herein be used to define the DKP algebra as a metric subalgebra (the notion to be explained below). This way of presenting the DKP algebra was first introduced in \cite{Schom} and it was later  applied to the study of concrete problems in the phase space picture of quantum mechanics \cite{M.Fernandes,M.Fernandes1,M.Fernandes2}. The relationship between  the DKP algebra and the Clifford algebra has been studied also in \cite{Sanch,Shim} but from  different points of view.

Our main result is stated through a theorem which presents the derivation of the DWH field equations from the action of the DKP algebra on subspaces of the original Clifford algebra. This formulation is similar to that of the algebraic spinors \cite{Riesz,Cartan} in which the first-order Dirac operator acts on the minimal left/right ideals of the Clifford algebra \cite{Schom,Hes,BenTucker}. Minimal ideals of an algebra are representation spaces for irreducible representations of the algebra on itself. Let us emphasize, however, that the results reported here refer to the DKP algebra and the classical DWH theory.

The sequence of the presentation is as follows. In Section 2, the Clifford algebra is defined in terms of a pair of Lagrangian subspaces\footnote{See footnote{}$^3$.} of a vector space endowed with a split bilinear form. This is the starting point to access the decomposition of the algebra into its subalgebras. This procedure leads to the framework of projection operators whose algebraic properties will be used in most calculations in the paper. The decomposition of the algebra by means of those projectors has been applied already in \cite{Bergdolt1} in order to recover the representation of Cartan's spin matrices and to the definition of spinors. The same technique was further applied in \cite{Bergdolt2} to the complete reduction of Clifford algebras.

In Section~3, a basic idempotent in the algebra is selected by means of automorphisms of the algebra. This will bring new insights into its invariance properties.

Section~4 introduces the ring of endomorphisms of the subspaces built from the basic invariant idempotent. This is then extended to multilinear endomorphisms in the Appendix following closely the original presentation in \cite{Schom}.  A new notation feature is also developed in Section 4. The subsequent use of this notation will serve two purposes: to make explicit how the algebra acts on its projected subspaces and/or ideals, and to provide a consistent system for calculations.

In Section~5, we introduce an additional structure into the algebra by equipping the original Langrangian subspaces with a metric tensor. This extra structure allows the construction of metric sub-algebras in the  Clifford algebra. The DKP algebra is introduced in this section in Proposition 1.

Thereafter  the action of the DKP algebra on a distinct class of subspaces of the Clifford algebra is set forth in  Proposition~2 of Section~6. Such an action encodes the \textit{linear} $k$-symplectic \cite{Igor1,Awa,MLeon} structure in the DKP formulation of the DWH theory in matrix representation-free form.

In Section~7, we globalize the algebras by means of an  algebra bundle description of the DKP algebra. This step allows us to construct a ``prototype'' for the DWH equations which is the first step towards the DWH theory.

In Section~8, the main result of the paper is stated in the theorem  which formulates the DWH field equations for antisymmetric fields in matrix representation-free form.

In Section~9, the construction of the  bracket consistent with the DKP algebraic representation of the DWH theory is developed in details. It is demonstrated that the matrix representation-free form of the bracket satisfy certain properties which generalize the properties of the standard Poisson bracket.

We finish the paper in Section 10 with the conclusions and the outline of further developments.

\section{Clifford Algebra of a Split Bilinear Form}
The Clifford algebra of the quadratic form $\mathcal{Q}$ on a vector space $W$ over a field $\mathbb{F}$ \footnote{In the applications considered here $\mathbb{F}$ can be assumed to be $\mathbb{R}$ or $\mathbb{C}$.}, char $\mathbb{F}\neq 2$,   is an associative algebra $\mathcal{C}$ with
$\mathbf{1}$ over $\mathbb{F}$, together with a linear injection
\begin{align*}
I:W \hookrightarrow \mathcal{C}
\end{align*}
such that
\begin{itemize}
  \item[(i)] for each $w$ in $W$, $w^2=\mathcal{Q}(w)\mathbf{1}$;
  \item[(ii)] the pair $\{\mathcal{C},I\}$ is universal, that is, given another pair $\{\mathcal{C}', W\stackrel{I'}{\longrightarrow}
      \mathcal{C}'\}$ which satisfies the condition (i), there is a unique algebra homomorphism $\psi$ from $\mathcal{C}$ to $\mathcal{C'}$ such that
      $I'=\psi\circ I.$
  \end{itemize}

Let $W$ be even-dimensional. A bilinear form $\langle \cdot, \cdot\rangle_W$ on $W$ is called \textit{split} if its \textit{Witt index}\footnote{The \textit{Witt index} of a non-degenerate symmetric bilinear form on a vector space $W$ is the dimension of a maximal totally isotropic subspace of $W$ relative to the form
$\langle \cdot, \cdot\rangle_W$ Maximal totally isotropic subspaces are also called \emph{Lagrangian subspaces} \cite{Lam, Mein}.} is $\frac{1}{2} \dim W$.
The Clifford algebra of a split bilinear form is described as follows. Let $V$ be a $n$-dimensional vector space and  $V^*$ its dual relative to the
canonical pairing $V\times V^*\stackrel{\langle \cdot, \cdot\rangle}\longrightarrow \mathbb{F}$, $(v,\alpha)\mapsto \langle \alpha,v\rangle$.  This pairing is
$\mathrm{GL}(n)$ invariant. We set $W$ as the direct sum $W=V\oplus V^{*}$. By using the $\mathrm{GL}(n)$ invariant $\langle \cdot,\cdot\rangle$ a split
non-degenerate symmetric bilinear form $\langle \cdot, \cdot \rangle_W: W\times W\rightarrow \mathbb{F}$ on $W$\footnote{A  \textit{split} bilinear form on $W$ implies that $W$ is hyperbolic \cite{Lam}.} can be defined as follows:
\begin{equation}
\langle w,w'\rangle_W:=\frac{1}{2}(\langle \alpha,v'\rangle + \langle \alpha',v\rangle), \label{bf1}
\end{equation}
where $w=v+\alpha$ and $w'=v'+ \alpha'$, with $v, v' \in V$ and $\alpha, \alpha' \in V^{*}$.
It is easily seen that the $\mathrm{GL}(n)$ invariant $\langle \cdot,\cdot\rangle $ turns out to be the quadratic form $\mathcal{Q}_{W}(w):=\langle w,w\rangle_W$ associated to this bilinear form. In other words, for $w=v+\alpha \in W$, $(w)^2=\mathcal{Q}_{W}(w)= \langle \alpha,v\rangle $ and
$2\langle w,w'\rangle_W=\mathcal{Q}_{W}(w+w')-\mathcal{Q}_{W}(w)-\mathcal{Q}_{W}(w')$. Note also that the subspaces $V$ and $V^*$ of $W$ are maximal totally isotropic (Lagrangian) because $\mathcal{Q}_{W}(v)=\mathcal{Q}_{W}(\alpha)=0$ for all $v \in V\hookrightarrow W$ and $\alpha \in V^* \hookrightarrow W$.

The Clifford algebra $\mathcal{C}(W,\mathcal{Q}_{W})$ of $(W,\mathcal{Q}_{W})$ can now be defined by setting the map  $I:W \hookrightarrow \mathcal{C}$ by the following relations:
 \begin{align}
v\alpha+\alpha v&=2\langle \alpha,v\rangle_W \cdot \mathbf{1}=\langle \alpha,v\rangle \!\cdot\!\mathbf{1}, \label{gn1}\\
vv'+v'v&=0, \label{gn2} \\
\alpha\alpha'+\alpha'\alpha&=0, \label{gn3}
\end{align}
for any choice of $v, v' \in V\hookrightarrow W$ and $\alpha, \alpha' \in V^*\hookrightarrow W$.

The set of relations (\ref{gn1})--(\ref{gn3}) is just another way of writing the traditional Clifford algebra relation
\begin{equation}
ww'+w'w=2\langle w,w'\rangle_W \!\cdot \!\mathbf{1}\label{saca}
\end{equation}
 associated to the bilinear form $\langle \cdot,\cdot\rangle_W$. This can be shown by computing the algebra product using the \emph{ordered pair} notation
 $(\cdot,\cdot)$ for the direct sum with the additional feature that the vectors are considered to be in the Clifford algebra. So $v\equiv(v, 0)$,
 $v'\equiv(v', 0)$, $\alpha\equiv(0, \alpha)$ and $\alpha'\equiv(0, \alpha')$ are elements of the Clifford algebra coming from $V$ and $V^*$
 respectively, image of the injection $I:W\rightarrow \mathcal{C}(W,\mathcal{Q}_{W})$ of the inclusions $i_1:V\rightarrow W$ and $i_2:V^*\rightarrow W$, and $0$ means the zero vector.  Thus the Clifford algebra relation (\ref{saca}) associated to the bilinear form (\ref{bf1}) implies
\begin{align*}
(v, 0)(v', 0)+(v', 0)(v, 0)=\langle 0,v'\rangle + \langle 0,v\rangle=0,  \\
(0, \alpha)(0, \alpha')+(0, \alpha')(0, \alpha)=\langle \alpha,0\rangle + \langle \alpha',0\rangle=0  \\
\mathrm{and} \;\;(v, 0)(0, \alpha)+(0, \alpha)(v, 0)=\langle \alpha,v\rangle,
 \end{align*}
 which are just the relations (\ref{gn1})--(\ref{gn3}).  Throughout this paper and for practical reasons we shall use notation (\ref{gn1})--(\ref{gn3}) in
 all calculations. So $vv'$ actually means $(v,0)(v',0)$, $v\alpha$ means $(v,0)(0,\alpha)$, $\alpha v$  means $(0,\alpha)(v,0)$ and
 $\alpha\alpha'$ means $(0,\alpha)(0,\alpha')$, for $v,v' \in V$ and $\alpha,\alpha' \in V^{*}$. This will bring much simplification in future
 calculations.

 The Clifford algebra $\mathcal{C}(W,\mathcal{Q}_{W})$ was called
 $\mathrm{G}_n$ in \cite{Schom} and here we will continue to adopt the same denomination.

\section{Invariant Projectors}
Let $\mathbf{T}$ be an element of $\mathrm{GL}(n)$, the general linear group of transformations on $V$. We consider the following automorphism
\begin{equation}
\Omega_{\mathbf{T}}(v) = \mathbf{T}v\quad \text{for all}\quad v\in V, \nonumber
\end{equation}
of the algebra $\mathrm{G}_n$.
It is easy to see that this automorphism $\label{aut1}$ preserves relations (\ref{gn1})--(\ref{gn3}) and hence extends to an automorphism of the whole
$\mathrm{G}_n$.
Among the elements $\mathbf{T}$ of the group we consider the scalings $c\mathbf{I}$ for $c$ a constant in $\mathbb{F}$. They give rise to dilation
automorphisms
$\omega_c:=\Omega_{c\mathbf{I}}$. They determine the identities
\begin{equation}
\omega_c{(v)}=cv \quad \mathrm{and} \quad \omega_c{(\alpha)}=c^{-1}\alpha \nonumber
\end{equation}
for all $v \in V$ and $\alpha \in  V^*$. Notice that if we choose the constants $c$ as pure imaginary, the $c\mathbf{I}$ become unitary automorphisms belonging to the circle group $\mathbb{T}$.

We are interested in the subalgebra of $\mathrm{G}_n$ which is invariant for all (dilation) automorphisms. Its definition is
\begin{equation}
\mathrm{G}_{n}^0:=\{\Lambda \in \mathrm{G}_n:\omega_{c}{\Lambda}=\Lambda \quad \text{for all} \quad c\in\mathbb{F}\}. \nonumber
\end{equation}
Let $v_1,\ldots,v_{m}$ be vectors in $V$ and $\alpha_1,\ldots,\alpha_{n}$ vectors in $V^*$. An element
\begin{equation}
\Gamma:=\alpha_1\cdots\alpha_{n}v_1\cdots v_{m} \nonumber
\end{equation}
in $\mathrm{G}_n$ satisfies $\omega_c(\Gamma)=c^{\,m-n}\Gamma$. In particular, $\Gamma$ belongs to the dilation (or unitary) invariant subalgebra
$\mathrm{G}_{n}^0$ when $m=n$. These particular elements will play a fundamental role in $\mathrm{G}_{n}$ and the invariant subalgebra $\mathrm{G}_{n}^0$ is
spanned by elements of such type.

For pairs of vectors $w=(v,0)$ and $w'=(0,\alpha)$ in $W$ we can build the following invariant projector:
\begin{equation}
\pi_W:=\frac{w'w}{\langle w',w\rangle_W}=\frac{(0,\alpha)(v,0)}{\langle \alpha,v\rangle_W} \equiv \frac{\alpha v}{\langle \alpha, v\rangle}.
\nonumber
\end{equation}
That this is a projection we can readily  verify:
\begin{align*}
\pi^2_W&=\frac{\alpha v}{\langle \alpha, v\rangle}\frac{\alpha v}{\langle \alpha, v\rangle}  \\
&=\frac{\langle \mathbf{\alpha},v \rangle \mathbf{\alpha}v}{\langle \alpha, v\rangle \langle \alpha, v\rangle}  \\
&= \pi_W.
\end{align*}
This idempotent element of the algebra is associated to a non-isotropic direction of the space $W$.

In all formulas referring to vector, covectors, and tensors in general, we shall use the kernel letters like $v$, $\alpha$, etc. and superscript or subscript running indices like $i,j,k,l,\ldots$. The convention to be adopted in this article is that Roman \emph{kernel} letters will always refer to contravariant tensor objects and Greek letters to covariant tensor objects. The superscript or subscript running indices are used in coherence with the standard principles of Ricci calculus where the positioning of the indices is according to the behavior under transformation. The summation convention on repeated indices is also implicit unless otherwise stated. We shall  warn the reader whenever new notation is introduced.

Let $e_1,\ldots , e_n$ be a basis of $V$ and $e^1,\ldots , e^n$ the dual basis in $V^*$. For each $j=1,\ldots , \dim V$ we introduce the following projectors:
\begin{align*}
{\mathcal{N}}_{j}=e_je^j \quad \text{(no summation).}
\end{align*}
 They give rise to the idempotent
\begin{equation}
P:= {\mathcal{N}}_{1}\cdot\cdot\cdot{\mathcal{N}}_{n}=P^2. \nonumber
\end{equation}
This idempotent is fundamental to what follows.
It satisfies the fundamental relations:
\begin{equation}
P(0,\rho)\equiv P\rho=0\quad \mathrm{and} \quad (v,0)P\equiv vP=0,\label{fr1}
\end{equation}
where $v=v^ie_i\in V$ and $\rho=\nu_je^j\in V^*$. It follows from the relations (\ref{fr1}) that contravariant tensors are left
zero divisors of the idempotent $P$ and covariant tensors are right zero divisors of $P$. This fact will be used systematically in most  forthcoming
calculations.

\section{Endomorphisms in $\mathrm{G}_n$ and Further Developments in the Notation}
 The algebraic calculus that will be used in future calculations will be developed in this and the next section. We start with the properties of the elements of
 the form $Pv$
 and $\beta P$ which we denote by

\begin{equation}
P_{v}:= Pv=P(v,0) \quad \text{and}\quad {}^{\beta}\!P:=\beta P=(0,\beta)P. \label{d}
\end{equation}
They satisfy the relations
\begin{gather}
P_{v}P =0,  \quad  P(^{\beta}\!P)=0, \quad    ^{\beta}\!P(^{\alpha}\!P)=0,  \quad  P_{v}P_{w}=0, \label{a00}\\
P_{v}(^{\beta}\!P) =\langle \beta,v\rangle(P), \quad ^{\beta}\!PP_{v} = \beta Pv \equiv ^{\beta}\!P_{v}, \label{d1a}
\end{gather}
which follow from using relation (\ref{gn1}) along with relations (\ref{fr1}). Hence, for the set of basis elements introduced at the end of the last section we have
\begin{align*}
P_{i} ({}^{j}\!P)\equiv P_{\!e_i} ({}^{e^j}\!\!P)&=\delta^j_iP ,  \\
{}^{e^j}\!\!PP_{\!e_i}&=e^jPe_i\equiv{}^j\!P_i.
\end{align*}
 From these algebraic properties follows
 \begin{equation}
 {}^j\!P_i({}^k\!P_l)=\delta^k_i\;{}^j\!P_l, \nonumber
 \end{equation}
 which shows that the ${}^j\!P_i$ form a set of linearly independent elements. These elements of $\mathrm{G}_n$ are operators acting from the left on
 objects of the form ${}^{\alpha}\!P$ and from the right on $P_{v}$, the action being simply the algebra multiplication. We denote
 the space of elements ${}^{\alpha}\!P$ and $P_{v}$ by $S^*$ and $S$, respectively.

  The projector $\Pi_{1}:=\sum_{i=1}^n{}^{i}\!P_{i}$ plays the role of
  the unit operator (the canonical unit tensor). We adopt the notation $\Pi_{0}\equiv P$. It follows that $\Pi_{0}\Pi_{1}=\Pi_{1}\Pi_{0}=0$.
Summing up, to every ordered pair $(\sigma,{v})$ in $V^{*}\times V$ there is a one-to-one correspondence with the elements ${}^{\sigma}\!P_{v}$ in
$\mathrm{G}_n$ and hence with the linear endomorphisms
\begin{gather*}
\End S^{\ast} \leftarrow V^{*} \times V \rightarrow \End S, \\
{}^{\sigma}\!P_{\!v}\testleft (\sigma,{v}) \mapsto {}^{\sigma}\!P_{\!v}, \\
\;\;\text{left action} \;\phantom{(\sigma,{v})}\; \text{right action}.
\end{gather*}
For example, we compute the left action on $S^*$,
\begin{align*}
 {}^{\sigma}\!P_{\!v}({}^{\alpha}\!P)&=\sigma P v\alpha P \nonumber \\
 &=\sigma P[\langle\alpha,v\rangle \cdot \mathbf{1}-\alpha v]P  \\
 &=\langle\alpha,v\rangle \; {}^{\sigma}\!P \cdot\mathbf{1}
 \end{align*}
 where we have used relations (\ref{gn1}) and (\ref{fr1}). The right action follows easily by analogous steps.

 These maps embed $V^{*}\times V$ into a matrix algebra of operators on the spaces $S^{*}$ and $S$. The algebraic product becomes the matrix product. The
 extension of  these algebraic properties is detailed in the Appendix.

\section{Metric Structure of the \text{DKP} Algebra}
\subsection{Additional Structure from a Metric on $V$}
 The introduction of a metric on $V$ corresponds to specifying the isomorphism $V \isom V^*$. Let  $g(\cdot|\cdot)_{V}$ be a symmetric non degenerate metric on the space~$V$ taking values in $\mathbb{F}$. The metric $g(\cdot|\cdot)_{V}$ gives rise to the flat and sharp isomorphisms
 \begin{align*}
 \flat:V&\rightarrow V^*,  \\
 v&\mapsto {}^\flat v = g(v|\cdot)_{V}:=\widetilde{v}, \quad v\in V,
 \end{align*}
 and $\sharp=\flat^{-1}$.

 The covector $\widetilde{v}$ in $V^*$ denotes the linear functional whose value at any vector $w \in V$ is
 \begin{equation}
 \widetilde{v}(w)=\langle \widetilde{v}, w\rangle=g(v,w)_{V}, \nonumber
 \end{equation}
where $\langle \cdot\,, \cdot\rangle$ continues to denote  the canonical pairing of vectors and covectors.

When the basis $\{e_i\}$ and its dual $\{e^i\}$ are being used, we write
$$
g(e_i,e_j):=g_{ij}
$$
as usual. So the metric $g$ has the expression $g=\sum_{ij}g_{ij}e^i\otimes e^j$. This symmetric tensor induces the isomorphism $\flat$ of the space $V$ with the $V^*$ so that each vector $v \in V$
goes into a linear function $g(v,\cdot)$ consisting of the scalar product with $v$.  In the case of $v\equiv e_i$ the map is given by
\begin{equation}
 e_k \mapsto \widetilde{e}_k=g_{ik}e^i \nonumber.
 \end{equation}
Moreover, this isomorphism induces a scalar product on the dual space $V^*$ and so a tensor of type ${2 \choose 0}$ which we will denote by $g^{-1}$. We write
$$
g^{-1}(e^i,e^j):=g^{ij}
$$
so that $g^{-1}=\sum_{ij}g^{ij}e_i\otimes e_j$. This metric induces a map into $V$:
\begin{equation}
 e^k\mapsto\tend{e}^k=g^{ik}e_i \nonumber.
 \end{equation}
The two maps --from $V$ to $V^*$ and vice versa-- are inverses of one another so that the operations of raising and lowering a given index are inverse
of each other,
\begin{equation}
g^{ij}g_{jk}= \delta^{i}_{k}.\nonumber
 \end{equation}
The ``matrices" $(g_{ij})$ and $(g^{ij})$ are inverses.

\subsection{DKP Algebras}

    To bring the metric into the algebra we build $(P_v)$ for $v=\tend{\alpha}$  that is
 \begin{equation}
 P_{\tend{\alpha}}:=P(\tend{\alpha},0)\equiv P\tend{\alpha}. \nonumber
 \end{equation}
Analogous relations to (\ref{d1a}) are found but with the inverse metric:
 \begin{gather}
P_{\tend{\alpha}}({}^{\beta}\!P)=g^{-1}(\alpha,\beta)P, \qquad  {}^{\sigma}\!P  P_{\tend{\alpha}}=\sigma P \tend{\alpha}. \label{d11}
\end{gather}
Clearly in the basis $\{e_i\}$, $\{e^i\}$ of $V$ and $V^*$ these relations are
 \begin{equation}
P_{\tend{e}^i}({}^{e^j}\!\!P)=P_{\tend{i}}({}^j\!P)=g^{ij}P, \qquad {}^{e^i}\!\!P  P_{\tend{e}^j}={}^i\!P_{\tend{j}}=g^{jk}({}^i\!P_k).
\label{d11b}
\end{equation}
 The DKP metric algebra can now be described:
\begin{prop}

 The set of elements
\begin{equation}
\mathbf{b}^{\alpha}= {}^{\alpha}\!P+P_{\tend{\alpha}} \nonumber
\end{equation}
in $\mathrm{G}_{n}$ generates a DKP-algebra.

 \begin{proof}In order to prove this proposition, we use relations (\ref{d11}), (\ref{a00}) and (\ref{d1a}) to obtain the product relation of the
generators of the DKP algebra but with the inverse metric. That is,
\begin{equation}
\mathbf{b}^{\alpha_{1}}\mathbf{b}^{\alpha_{2}}\mathbf{b}^{\alpha_{3}}+\mathbf{b}^{\alpha_{3}}\mathbf{b}^{\alpha_{2}}\mathbf{b}^{\alpha_{1}}=g^{-1}(\alpha_1,\alpha_2)
\mathbf{b}^{\alpha_{3}}+g^{-1}(\alpha_3,\alpha_2)\mathbf{b}^{\alpha_{1}}. \nonumber
\end{equation}
\end{proof}
\end{prop}
Note that, using the \emph{negative sign in front of the inverse metric}\footnote{The negative sign in front of the metric is sometimes used in the case of
the Clifford algebra relations in order to simplify the correspondence with quaternions: $vw+wv=-2g(v,w)$, with $g(v,w)$ the Euclidean metric.} we can
also define
 \begin{equation}
 \underline{\mathbf{b}}{}^{\alpha}:={}^{\alpha}\!P-P_{\tend{\alpha}}, \label{nmdkp}
 \end{equation}
 from which follows
 \begin{equation}
\underline{\mathbf{b}}{}^{\alpha_1}\underline{\mathbf{b}}{}^{\alpha_2}\underline{\mathbf{b}}{}^{\alpha_3}+\underline{\mathbf{b}}{}^{\alpha_3}\underline{\mathbf{b}}{}^{\alpha_2}\underline{\mathbf{b}}{}^{\alpha_1}=-g^{-1}(\alpha_1,\alpha_2)\underline{\mathbf{b}}{}^{\alpha_3}
-g^{-1}(\alpha_3,\alpha_2)\underline{\mathbf{b}}{}^{\alpha_1}. \nonumber
\end{equation}
There is also the reciprocal version that we write in the form,
\begin{equation}
\underline{\mathbf{b}}{}_{v}:=P_v-{}^{\widetilde{v}}\!P. \label{dkds}
\end{equation}
They satisfy the relations,
\begin{equation}
\underline{\mathbf{b}}{}_{v_1}\underline{\mathbf{b}}{}_{v_2}\underline{\mathbf{b}}{}_{v_3}+\underline{\mathbf{b}}{}_{v_3}\underline{\mathbf{b}}{}_{v_2}\underline{\mathbf{b}}{}_{v_1}=-g(v_1,v_2)\underline{\mathbf{b}}{}_{v_3}
-g(v_3,v_2)\underline{\mathbf{b}}{}_{v_1}. \nonumber
\end{equation}

The most common form found in the literature \cite{Corson} for the DKP generators can also be rewritten in this algebraic setting, they are
\begin{equation}
\beta_i:= P_i+{}^{\widetilde{i}}\!P,\label{bddk}
\end{equation}
which satisfy
\begin{equation}
\beta_i\beta_j\beta_k+\beta_k\beta_j\beta_i=g_{ij}\beta_k+g_{kj}\beta_i. \nonumber
\end{equation}

The identity element of the algebra is the idempotent
\begin{equation}
1_{\mathrm{DKP}}=P+\sum_{i=1}^n {}^i\!P_i. \nonumber
\end{equation}
The generators (\ref{bddk}) can also be written with the negative sign in front of the metric, that is,
\begin{equation}
\underline{\beta}{}_i:= Pe_i- g_{ij}e^jP,\label{ndkp}
\end{equation}
which leads to
\begin{equation}
\underline{\beta}{}_i\underline{\beta}{}_j\underline{\beta}{}_k+\underline{\beta}{}_k\underline{\beta}{}_j\underline{\beta}{}_i=-g_{ij}\underline{\beta}{}_k-g_{kj}\underline{\beta}{}_i.
\nonumber
\end{equation}
The form (\ref{ndkp}) of the DKP generators  was already used in \cite{M.Fernandes} in connection with the quantum Liouville equation in phase space.

\section{Representation on Invariant Subspaces}
In what follows we will name each invariant subspace using the  index ``$p$" defined in the Appendix. This will be better explained with the development of this section.
\subsection{Projecting onto Antisymmetric Tensors}
 Let us use the projections $\Pi_p$ defined in Eqs. (\ref{pro}), (\ref{gips}) of the Appendix to project $\mathrm{G}_{n}$ as follows:
\begin{equation}
\mathcal{Z}_{(p)}:=\Pi_{0}\mathrm{G}_{n}\Pi_p+\Pi_{1}\mathrm{G}_{n}\Pi_0\mathrm{G}_{n}\Pi_{p}. \nonumber
\end{equation}
Elements of this space have the general form
 \begin{equation}
 Z_{(p)}= P_{\!\!\!\underbrace{\rule{0in}{0ex}^{\bullet \ldots \bullet}}_{p\;\mathrm{vectors}}} + {}^{\bullet}\!P_{\!\!\underbrace{\rule{0in}{0ex}^{\bullet \ldots \bullet}}_{p\;\mathrm{vectors}}}\!\!\!.\nonumber
 \end{equation}
By using the multi-index notation they can be written as
\begin{equation}
  Z_{(p)}= P_{I}+{}^{\gamma}\!P_{I}, \nonumber
  \end{equation}
with $p=|I|$ running from $0$ to $n=\dim V$ and $I$ accounting for a multi-index formed by a set of indices which are antisymmetric by permutation of
pairs. Objects with the algebraic properties of  $P_{I}$ are known as antisymmetric tensors of rank ``p".

The dimension of the space $\mathcal{Z}_{(p)}$ is
$(n!(n+1))/(p!(n-p)!)$.
\begin{prop}

The $\underline{\mathbf{b}}{}^{\alpha}$-generators of the DKP algebra act from the left on the space $\mathcal{Z}_{(p)}$:
\begin{equation}
\mathrm{DKP} \times \mathcal{Z}_{(p)} \rightarrow \mathcal{Z}_{(p)}:
( \underline{\mathbf{b}}{}^{\alpha},Z_{(p)})\mapsto \underline{\mathbf{b}}{}^{\alpha}Z_{(p)}, \nonumber
\end{equation}
$\underline{\mathbf{b}}{}^{\alpha}Z_{(p)} \in \mathcal{Z}_{(p)}$.

\begin{proof} One can easily check  from the idempotency of $P$ that ${}^{\alpha}\!PP_{I}={}^{\alpha}\!P_{I}$, and  using relations
(\ref{a00}) along with  $P_{\tend{\alpha}}({}^{\gamma}\!P)=g^{-1}(\alpha,\gamma) P$  we compute,
\begin{align*}
\underline{\mathbf{b}}{}^{\alpha}Z_{(p)}&=\underline{\mathbf{b}}{}^{\alpha}[P_{I}+{}^{\gamma}\!P_{I}]  \\
&= [{}^{\alpha}\!P - P_{\tend{\alpha}}][P_{I}+{}^{\gamma}\!P_{I}] \\
&= {}^{\alpha}\!P_{I}-g^{-1}(\alpha,\gamma)P_I,
\end{align*}
which is again an element of $\mathcal{Z}_{(p)}$. The algebra is associative and it is easily seen that $1_{\mathrm{DKP}}Z_{(p)}=Z_{(p)}$.
\end{proof}
\end{prop}
 Some particular cases will be presented next.
\subsection{Scalars and Covectors}
In this case the operator (\ref{nmdkp}) acts on the subspace $\mathcal{Z}_{(p=0)}=\Pi_{0}\mathrm{G}_{n}\Pi_0+\Pi_{1}\mathrm{G}_{n}\Pi_0\mathrm{G}_{n}\Pi_{0}\equiv \Pi_{0}\mathrm{G}_{n}\Pi_0+\Pi_{1}\mathrm{G}_{n}\Pi_0$.
\begin{align*}
\underline{\mathbf{b}}{}^{\alpha}Z_{(0)}&=\underline{\mathbf{b}}{}^{\alpha}[s(P)+{}^{\gamma}\!P]  \\
&= [{}^{\alpha}\!P - P_{\tend{\alpha}}][s(P)+{}^{\gamma}\!P] \\
&= s({}^{\alpha}\!P)-g^{-1}(\alpha,\gamma)P,
\end{align*}
where $s \in \mathbb{F}$ and $\gamma \in V^*$.
\subsection{Vectors and $(1,1)$-tensors}
 This is the case $p=1$. The operator (\ref{nmdkp}) acts on the subspace
\begin{equation}
\mathcal{Z}_{(1)}=\Pi_{0}\mathrm{G}_{n}\Pi_1+\Pi_1\mathrm{G}_{n}\Pi_0\mathrm{G}_{n}\Pi_{1}. \nonumber
\end{equation}
For $\alpha$ an element of $V^{*}$ and $v$, $w$ any elements of $V$,
a typical element of the space $\mathcal{Z}_{(1)}$ has the form
\begin{equation}
P_{v}+  {}^{\alpha}\!P_{w}, \nonumber
\end{equation}
which in the basis $\{e_1,\ldots, e_n\}$ of $V$ and $\{e^1,\ldots, e^n\}$ of $V^*$ reads
\begin{equation}
Z_{(1)}=v^i(P_{i})+T^{\;k}_{l}({}^{l}\!P_{k}),\nonumber
\end{equation}
where $P_i=Pe_i$ and ${}^{l}\!P_{k}=e^l(P)e_k$.
%$\textbf{Proposition 3.}$:
%\end{proposition}
 Once again, we compute the algebra multiplication
\begin{align*}
\underline{\mathbf{b}}{}^{\alpha}Z&=[{}^{\alpha}\!P - P_{\tend{\alpha}}][P_{v}+  {}^{\beta}\!P_{w}]  \\
&= {}^{\alpha}\!P_{v}-g^{-1}(\alpha,\beta)P_w
\end{align*}
which again belongs to $\mathcal{Z}_{(1)}$.  All the
computations follow from relations (\ref{a00}) and (\ref{d11}).

The next invariant subspace, $p=2$, is  the space of pairs comprised of ${2 \choose 0}$ and ${2 \choose 1}$ tensors. Elements describing these pairs are written
  in the algebra in the form $Z_{(2)}= P_{vw}+^{\alpha}P_{xy}$. Note that $P_{vw}=-P_{wv}$ and $^{\alpha}P_{xy}=-^{\alpha}P_{yx}$ as a
  consequence of the  antisymmetric nature of product in their definitions.

  Similar results follow for antisymmetric tensors up to the rank $p=n$.
\section{DWH Tensor Field Theory and DKP Algebra}
In this section we will show how the actions of the DKP algebra give rise to Hamiltonian equations of the same form as the ones of the DWH theory.
To this purpose we have to globalize the algebras.

Suppose $E_1\overset{\pi_1}{\longrightarrow}M$ is a vector bundle over a $n$-dimensional Euclidean manifold $M$. Locally the bundle can be identified with $M\times W$ where all fibers are identified with the fixed vector space $W=V\oplus V^*$ introduced in Section 2. We
assume once again that $W$ is endowed with the bilinear form (\ref{bf1}) and its corresponding quadratic form $(w)^2=\mathcal{Q}_{W}[w]$. At a point $q\in
M$, this quadratic form on the fiber $W_q=\pi_1^{-1}(q)$ can be used to construct the Clifford algebra $\mathrm{G}_n(W_q)$. The result is the Clifford
bundle $\mathrm{G}_n(E_1)\rightarrow M$ of $E_1$. We call such an algebra bundle $\mathrm{G}_{(M,W)}$,
\begin{equation}
\mathrm{G}_{(M,W)}:=\bigcup_{q\in M}\{q\}\times\text{G}_n(W_q). \nonumber
\end{equation}
 Since $\mathrm{DKP}\subset \mathrm{G}_{n}$, the $\mathrm{DKP}_{(M,W)}$ bundle  is under consideration too.

Let $\psi$ be a section of $\mathrm{G}_{(M,W)}$ with values in the subspaces $\mathcal{Z}_{(p)}$ of $\mathrm{G}_{(M,W)}$. Let
$\Gamma(\mathcal{Z}_{(p)})$ denote the space of these sections. Thus $\psi$ is  a map $M\rightarrow \mathcal{Z}_{(p)}$, such that $\psi(q) \in
\mathcal{Z}_{(p)}(q)$ for all $q\in M$, where $\mathcal{Z}_{(p)}(q)\subset \mathrm{G}_{(M,W)}(q)$ denotes the vector-space fiber of
$\mathcal{Z}_{(p)}$ above the point $q$. The DKP operators act on these fibers.

Let $\mathfrak{F}(M)$ denote the algebra (over real numbers) of all real-valued, $\mathrm{C}^{\infty}$ functions on $M$. The space
$\Gamma(\mathcal{Z}_{(p)})$ is an $\mathfrak{F}(M)$-module. The product $f\psi$ is defined point-wise:
\begin{equation}
(f\psi)(q)=f(q)\psi(q) \nonumber
\end{equation}
for $q \in M$.
Let $\psi_a$, $1\leq a \leq \dim \mathcal{Z}_{(p)}(q)$ be an $\mathfrak{F}(M)$-basis of $\Gamma(\mathcal{Z}_{(p)})$. Each $\psi \in
\Gamma(\mathcal{Z}_{(p)})$ can be written in the form $\psi=f^a\psi_a$ with $f^a \in \mathfrak{F}(M)$.

Now we describe the action of the DKP operators. For $q \in M$, let $L(\mathcal{Z}_{(p)}(q))$ denote the  algebra of $\mathbb{R}$-linear maps
\begin{equation}
\underline{\mathbf{b}}:\mathcal{Z}_{(p)}(q)\rightarrow \mathcal{Z}_{(p)}(q). \nonumber
\end{equation}
according to what has been established in Proposition~2 of the  Section~6. As~$q$ varies, we obtain an algebra bundle, denoted by $L(\mathcal{Z}_{(p)})$. Let $\Gamma(L(\mathcal{Z}_{(p)}))$ denote the space of sections of this algebra bundle. An element $\underline{\mathbf{b}} \in \Gamma(L(\mathcal{Z}_{(p)}))$ has a value, at each  point $q \in M$, which is an $\mathbb{R}$-linear map
$$
\underline{\mathbf{b}}(q): \mathcal{Z}_{(p)}(q)\rightarrow \mathcal{Z}_{(p)}(q).
$$
Two elements, $\underline{\mathbf{b}}$, $\underline{\mathbf{b}}^{\prime}$ of $\Gamma(L(\mathcal{Z}_{(p)}))$ can be multiplied, that is,
\begin{equation}
(\underline{\mathbf{b}}\,\underline{\mathbf{b}}^{\prime})(q)=\underline{\mathbf{b}}(q)\underline{\mathbf{b}}^{\prime}(q) \nonumber
\end{equation}
for $q \in M$. This makes $\Gamma(L(\mathcal{Z}_{(p)}))$ an algebra over $\mathbb{F}$. Therefore a section $\psi \in \Gamma(\mathcal{Z}_{(p)})$ can be
multiplied by a $\underline{\mathbf{b}} \in \Gamma(L(\mathcal{Z}_{(p)}))$ as follows:
\begin{equation}
(\underline{\mathbf{b}}\psi)(q)=\underline{\mathbf{b}}(q)\psi(q) \nonumber
\end{equation}
for $q \in M$. In this way, $\Gamma(\mathcal{Z}_{(p)})$ becomes an $\Gamma(L(\mathcal{Z}_{(p)}))$-module.

It is convenient for the purpose of the formulation of the DKP description of the DWH theory to be able to write  these algebraic operations in terms of
tensor components.

In Section 6 we classified the invariant subspaces of the algebra $\mathrm{G}_{n}$ according to the index  ``$p$" which describes the rank of the
tensor objects in the algebra. Let us start with $p=0$. This corresponds to the description of a single  scalar field variable ``$y$". We denote
$(y,\pi_i)$ the tensor field variables. These variables will be represented as  elements of the space $\mathcal{Z}_{(0)}(q) \subset \mathrm{G}_{n}(W_q)$.
In this way, to the pair of field variable $(y,\pi_i)$ we assign the element
 \begin{equation}
 \mathcal{\psi}_{(0)}= y(P)+\pi_{i}({}^{i}\!P) \label{mas1}
 \end{equation}
of $\Gamma(\mathcal{Z}_{(p=0)})$, which assumes values in the fiber $\mathcal{Z}_{(0)}(q)$.

The reader should take careful notice of how we position the indices.
The positioning of covariant and contravariant indices is critical because these indices have to be set in consistence with the elements of the algebra on
which the DKP generators will act. Therefore their positioning is not directly following the notation conventions of the standard DWH theory at this
stage. In the DWH theory the field indices are generic indices commonly written with Roman letters. They  refer to coordinates in the space of field
variables.

Next we introduce the differential operator:

\begin{remark}
We denote $\nabla_{\!\!\mathcal{Z}_{(0)}}$ the differential operator
\begin{equation}
(P)\frac{\partial }{\partial y}+({}^{j}\!P)\frac{\partial }{\partial \pi^{j}} \nonumber
\end{equation}
which acts  as derivations on the functions on $M\times \mathcal{Z}_{(0)}$.
\end{remark}
From these settings, a prototype for the DWH equations for the fields of the form (\ref{mas1}) emerges:
\begin{equation}
\underline{\mathbf{b}}{}^{a}\partial_a \phi_0=\nabla_{\!\!\mathcal{Z}_{(0)}} \mathcal{H}, \label{s11}
\end{equation}
where $\partial_a$, $a\!=\!1,\ldots,n$ is the standard basis vector fields induced by the choice of coordinates on $M$.

 Let us check this claim.
The action on the l.h.s. of (\ref{s11}) is as follows
\begin{align}
\underline{\mathbf{b}}{}^{a}\partial_a \phi_0&=\underline{\mathbf{b}}{}^{a}\partial_{a}\{y(P)+\pi_{i}({}^{i}\!P)\} \nonumber \\
&=\{{}^{a}\!P - \delta^{ab}(P_{b})\}\{\partial_{a}y(P)+\partial_{a}\pi_{i}({}^{i}\!P)\} \nonumber \\
&=\partial_{a}y({}^{a}\!P)-\delta^{ab} \delta_{b}^{i} \partial_{a}\pi_{i}(P) \nonumber \\
&=\partial_{a}y({}^{a}\!P)- \partial_{a}\pi^{a}(P). \label{lhw}
 \end{align}
 The r.h.s. of (\ref{s11}) is
\begin{equation}
\nabla_{\!\!\mathcal{Z}_{(0)}} \mathcal{H}= (P)\frac{\partial \mathcal{H}}{\partial y}+
          ({}^{a}\!P)\frac{\partial \mathcal{H}}{\partial \pi^{a}}.\label{srw}
\end{equation}
 We equate   (\ref{lhw}) to (\ref{srw}) to obtain,
\begin{equation}
\partial_{a}\pi^a=-\frac{\partial \mathcal{H}}{\partial y} \nonumber
\end{equation}
and
\begin{equation}
\partial_{a}y=\frac{\partial \mathcal{H}}{\partial \pi^{a}}. \nonumber
\end{equation}

The next case is $p=1$. In this case, the pair  $(y^{a},\pi^{\;a}_{b})$ of field variables in $\Gamma(\mathcal{Z}_{(1)})$ is accordingly set as
\begin{equation}
         \psi_{(1)}= y^{a}(P_{\!\!a})+\pi_{b}^{\;a}({}^{b}\!P_{\!\!a}) \nonumber
 \end{equation}
and the operator
\begin{equation}
\nabla_{\mathcal{Z}_{(1)}}:=(P_{\!i})\frac{\partial }{\partial y_{i}}+
          ({}^{j}\!P_{\!i})\frac{\partial }{\partial \pi^{j}_{\;\,i}}. \nonumber
\end{equation}
The DWH-like  equations are
        \begin{equation}
       \underline{\mathbf{b}}{}^{c}\partial_{c}\psi_{(1)}=\nabla_{\!\!\mathcal{Z}_{(1)}} \mathcal{H}. \label{YYZ}
        \end{equation}
The action on the l.h.s. of (\ref{YYZ}) is as follows:
\begin{align}
\underline{\mathbf{b}}{}^{c}\partial_{c}\psi_{(1)}&=\underline{\mathbf{b}}{}^{c}\partial_{c}\{y^{a}(P_{\!\!a})+\pi^{\;a}_{b}({}^{b}\!P_{\!\!a})\}
\nonumber \\
&=\{({}^{c}\!P) - \delta^{cd}(P_{\!d})\}\{\partial_{\mu}y^{a}(P_{\!\!a})+\partial_{c}\pi^{\;a}_{b}({}^{b}\!P_{\!\!a})\} \nonumber \\
&= \partial_{c}y^{a}({}^{c}\!P_{\!\!a})-\; \partial_{c}\pi^{ca}(P_{\!\!a}). \label{lhs1}
\end{align}
  The r.h.s. of (\ref{YYZ}) is
\begin{equation}
\nabla_{\!\!\mathcal{Z}_{(1)}} \mathcal{H}= (P_{\!\!a})\frac{\partial\mathcal{ H}}{\partial y_{a}}+
          ({}^{c}\!P_{\!\!a})\frac{\partial \mathcal{H}}{\partial \pi^{c}_{\;a}}.\label{rhs2}
\end{equation}
 By equating  (\ref{lhs1}) to (\ref{rhs2}) as proclaimed by (\ref{YYZ}) we obtain:
\begin{equation}
\partial_{c}\pi^{ca}=-\frac{\partial \mathcal{H}}{\partial y_{a}}\Rightarrow \partial_{c}\pi^{c}_{\;a}=-\frac{\partial \mathcal{H}}{\partial y^{a}}
\nonumber
\end{equation}
and
\begin{equation}
\partial_{c}\,y^{a}=\frac{\partial \mathcal{\mathcal{H}}}{\partial \pi^{c}_{\;a}}. \nonumber
\end{equation}
We have used the Euclidean metric to lower and raise the indices.
These DWH-like equations for the Hamiltonian $\mathcal{H}$ describe the vector field $\partial_{c}$ coupled to the
operator $\underline{\mathbf{b}}{}^{c}$ that represents the k-symplectic structure. In order to retrieve the DHW theory from this construction we have to
turn to  the spacetime vector field $\partial_{\mu}$ and use the operator $\underline{\mathbf{\beta}}{}^{\mu}\partial_\mu$ instead of
$\underline{\mathbf{b}}{}^{c}\partial_{c}$. In other words we have to change from the Euclidean frame $\partial_a$ to the spacetime frame
$\partial_{\mu}$. This will be done in the next section. The operators $\underline{\mathbf{\beta}}{}^{\mu}$ were used in \cite{Igor1} to find the De
Donder-Weyl Hamiltonian field equations in matrix form. They were required to satisfy the relation
\begin{equation}
\underline{\mathbf{\beta}}{}^{\mu}\underline{\mathbf{\beta}}{}^{\nu}\underline{\mathbf{\beta}}{}^{\gamma}+
\underline{\mathbf{\beta}}{}^{\gamma}\underline{\mathbf{\beta}}{}^{\nu}\underline{\mathbf{\beta}}{}^{\mu}=-\delta^{\mu\nu}\underline{\mathbf{\beta}}{}^{\gamma}
-\delta^{\gamma\nu}\underline{\mathbf{\beta}}{}^{\mu}\label{ndkc}
\end{equation}
which is not the form usually  found in particle Physics in the DKP theory. The minus signs in the r.h.s. of (\ref{ndkc}) are crucial in order to represent the
k-symplectic structure. In Section 5 we clarified the origin of these minus signs in the construction of the $\underline{\mathbf{\beta}}$'s.

\section{DWH Theory for Antisymmetric Fields Arising out of the DKP Algebra}
Let $X$ denote the Minkowski spacetime manifold. Coordinates of a point $x\in X$ will be denoted by $x^{\mu}$, $\mu=0,\ldots,3$.

In order to obtain the DWH theory for antisymmetric fields we will  transfer the bundle structures described in Section~7 to spacetime.

Let $\mathcal{Z}_{(p)}\stackrel{\pi}{\longrightarrow} M$ be the vector bundle with fiber $\mathcal{Z}_{(p)}(q)$ as described in Section~7. Let
$h: X\rightarrow M$ be a continuous map of the spacetime $X$ into $M$, $\text{dim M}=4$. Consider the set
\begin{equation}
h^*\mathcal{Z}_{(p)}\equiv \{(x,Z)\in X\times \mathcal{Z}_{(p)}:h(x)=\pi(Z)\} \nonumber
\end{equation}
of points in the product $X\times \mathcal{Z}_{(p)}$.
Now, any section (field) of $\mathcal{Z}_{(p)}\stackrel{\pi}{\longrightarrow} M$ induces a section of $h^*\mathcal{Z}_{(p)}$ called the pullback section
$h^*s:=s\circ h$. In this way,  the DWH equations for fields we described in Section~6.3 can be written for  fields over the spacetime
$X$. The result as it will be shown shortly is the DWH theory for antisymmetric fields.

We claim that the DWH equations for antisymmetric fields of rank $p$ arise out of the equation
 \begin{equation}
 \underline{\mathbf{\beta}}{}^{\mu}\partial_\mu \Psi_{(p)}=\nabla_{\!\!\mathcal{Z}_{(p)}} H. \label{s1}
 \end{equation}
 In this equation, the first order differential operator $\underline{\mathbf{\beta}}{}^{\mu}\partial_\mu$ will now act on sections of
 $h^*\mathcal{Z}_{(p)}$. Therefore, $\Gamma(\mathcal{Z}_{(p)})$ becomes a $\mathfrak{F}(X)$-module via pullback $h^*\psi_{(p)}$.

 Let us prove this claim. Starting with the case $p=0$ we write
 \begin{align*}
\Psi_{0}&= h^*y(P)+h^*\pi_i({}^{i}\!P)  \\
 &=\textsf{y}(P)+\uppi_i({}^{i}\!P),
 \end{align*}
 where we have introduced the notations: $h^*\psi_{(0)}\equiv\Psi_{0}$, $h^*y \equiv \textsf{y}$ and $h^*\pi_i\equiv \uppi_i$ for the pullbacks.
After choosing local charts around $q\in M$ as well as around $x\in X$, we denote $\Lambda^{\;\mu}_{a}\partial_{\mu}$ the pull back
$h^*(\partial_a)$ of the
standard basis  vector field $\partial_a$ on $M$ written in the spacetime basis vector field  $\partial_{\mu}$ on $X$.
We make the transformation  $\underline{\mathbf{b}}{}^{a}\partial_{a}\rightarrow \underline{\mathbf{b}}{}^{a}\!\Lambda^{\;\mu}_{a}\partial_{\mu}$. Notice
that the contraction $\underline{\mathbf{b}}{}^{a}\!\Lambda^{\;\mu}_{a}$,  with $\underline{\mathbf{b}}{}^{a}=({}^{a}\!P) - \delta^{ab}(P_{b})$, satisfies
the relations (\ref{ndkc}) with the identification $\underline{\mathbf{b}}{}^{a}\!\Lambda^{\;\mu}_{a} \equiv \underline{\mathbf{\beta}}{}^{\mu}$.

 For $p=0$, the action on the l.h.s. of (\ref{s1}) is as follows
\begin{align}
\underline{\mathbf{\beta}}{}^{\mu}\partial_\mu
\Psi_{0}&=\underline{\mathbf{b}}{}^{a}\!\Lambda^{\;\mu}_{a}\partial_{\mu}\{\textsf{y}(P)+\uppi_{i}({}^{i}\!P)\} \nonumber \\
&=\{({}^{a}\!P) - \delta^{ab}(P_{b})\}\Lambda^{\;\mu}_{a}\{\partial_{\mu}\textsf{y}(P)+\partial_{\mu}\uppi_{i}({}^{i}\!P)\} \nonumber \\
&=\Lambda^{\;\mu}_{a}\partial_{\mu}\textsf{y} ({}^{a}\!P)-\partial_{\mu}\Lambda^{\;\mu}_{a}\uppi^a(P), \label{lhss}
\end{align}
 and the r.h.s. of (\ref{s1}) explicitly is
\begin{equation}
\nabla_{\!\!\mathcal{Z}_{(0)}} H= (P)\frac{\partial H}{\partial \textsf{y}}+
          ({}^{a}\!P)\frac{\partial H}{\partial \uppi^{a}}.\label{srh}
\end{equation}
 By equating   (\ref{lhss}) to (\ref{srh}) as claimed we obtain
\begin{equation}
\partial_{\mu}[\Lambda^{\;\mu}_{a}\uppi^a]=-\frac{\partial H}{\partial \textsf{y}} \nonumber
\end{equation}
and
\begin{equation}
\Lambda^{\;\mu}_{a}\partial_{\mu}\textsf{y}=\frac{\partial H}{\partial \uppi^{a}}\Rightarrow \partial_{\mu}\textsf{y}=\frac{\partial
H}{\partial[\Lambda^{\;\mu}_{a}\uppi^{a}]}. \nonumber
\end{equation}
These equations suggest that we can chose $\textsf{p}^{\mu}:=\Lambda^{\;\mu}_{a}\uppi^a$ to represent the polymomenta of the DWH  theory. Therefore we can think of  the transformation~$\Lambda$ between coordinate charts on $X$ as the map from the field variables $\uppi^a$ in the pullback of the $\mathfrak{F}(M)$-module
$\Gamma(\mathcal{Z}_{(p)})$ to the polymomenta of the DWH theory.

The next case is $p=1$. In this case, the pair  $(\textsf{y}^{a},\uppi^{\;a}_{b})$ of field variables in $\mathrm{DKP}_{(M,W)}$ is accordingly
set as
\begin{equation}
         \Psi_{(1)}= \textsf{y}^{a}(P_{\!\!a})+\uppi_{b}^{\;a}({}^{b}\!P_{\!\!a}) \nonumber
 \end{equation}
and the operator
\begin{equation}
\nabla_{\!\!\mathcal{Z}_{(1)}}:=(P_{\!i})\frac{\partial }{\partial \textsf{y}_{i}}+
          ({}^{j}\!P_{\!i})\frac{\partial }{\partial \uppi^{j}_{\;i}}. \nonumber
\end{equation}
The DWH  equations are
        \begin{equation}
       \underline{\mathbf{\beta}}{}^{\mu}\partial_\mu \Psi_{(1)}=\nabla_{\!\!\mathcal{Z}_{(1)}} H. \label{q1a}
        \end{equation}
The action on the l.h.s. of (\ref{q1a}) can be computed:
\begin{align}
\underline{\mathbf{\beta}}{}^{\mu}\partial_\mu\Psi_{(1)}&=\Lambda^{\mu}_{\;c}\underline{\mathbf{b}}{}^{c}\partial_{\mu}\{\textsf{y}^{a}(P_{\!\!a})+\uppi^{\;a}_{b}({}^{b}\!P_{\!\!a})\}
\nonumber \\
&=\Lambda^{\mu}_{\;c}\{({}^{c}\!P) - \delta^{cd}(P_{\!d})\}\{\partial_{\mu}\textsf{y}^{a}(P_{\!\!a})+\partial_{\mu}\uppi^{\;a}_{b}({}^{b}\!P_{\!\!a})\}
\nonumber \\
&= \Lambda^{\mu}_{\;c}\partial_{\mu}\textsf{y}^{a}({}^{c}\!P_{\!\!a})-\Lambda^{\mu}_{\;c}\; \partial_{\mu}\uppi^{ca}(P_{\!\!a})  \label{lhs}
\end{align}
 and the r.h.s is
\begin{equation}
\nabla_{\!\!\mathcal{Z}_{(1)}} H= (P_{\!\!a})\frac{\partial H}{\partial \textsf{y}_{a}}+
          ({}^{c}\!P_{\!\!a})\frac{\partial H}{\partial \uppi^{c}_{\;a}}.\label{rhs}
\end{equation}
 By equating  (\ref{lhs}) to (\ref{rhs}) as declared in (\ref{q1a}), we obtain:
\begin{equation}
\partial_{\mu}[\Lambda_{\;c}^{\mu}\uppi^{ca}]=-\frac{\partial H}{\partial \textsf{y}_{a}}\Rightarrow \partial_{\mu}\textsf{p}^{\mu a}=-\frac{\partial
H}{\partial \textsf{y}_{a}}\Rightarrow  \partial_{\mu}\textsf{p}^{\mu}_{\;a}=-\frac{\partial H}{\partial \textsf{y}^{a}} \nonumber
\end{equation}
and
\begin{equation}
\Lambda^{\mu}_{\;c}\partial_{\mu}\textsf{y}^{a}=\frac{\partial H}{\partial \uppi^{c}_{\;a}}\Rightarrow
\partial_{\mu}\textsf{y}^{\;a}=\frac{\partial H}{\partial (\Lambda^{\mu}_{\;c}\uppi^{c}_{\;a})}\Rightarrow
\partial_{\mu}\textsf{y}^{a}=\frac{\partial H}{\partial \textsf{p}^{\mu}_{\;a}}. \nonumber
\end{equation}
We have used the Euclidean metric to lower and raise the Latin indices. These equations match the form of the DWH equations for the Hamiltonian~$H$. Thus, in general for all values of $``p"$ we state these results as follows.

\begin{definition}
We set the operator
\begin{equation}
\nabla_{\!\!\mathcal{Z}_{(p)}}:=(P_{I})\frac{\partial }{\partial \textsf{y}_{I}}+
          ({}^{a}\!P_{I})\frac{\partial }{\partial \uppi^{a}_{\;I}}. \label{dzop}
\end{equation}
\end{definition}
 This operator acts  as derivations on the functions on $X\times \mathcal{Z}_{(p)}$. Because the derivations $\frac{\partial }{\partial
 \textsf{y}_{I}}$ and $\frac{\partial }{\partial \pi^{a}_{\;I}}$ are ``coefficients" of $(P_{I})$ and $({}^{\alpha}\!P_{I})$,
 $\nabla_{\!\!\mathcal{Z}_{(p)}}F$ is an element of the invariant subspace $\mathcal{Z}_{(p)}$ for a function $F$ on $X\times \mathcal{Z}_{(p)}$. From the settings,
$\Psi_{(p)}=\textsf{y}^{I}(P_{I})+\uppi^{\;I}_{a}({}^{a}\!P_{I})=\textsf{y}^{j_1,\ldots,j_p}(P_{j_1,\ldots,j_p})+\uppi^{\;j_1,\ldots,j_p}_{a}({}^{a}\!P_{j_1,\ldots,j_p})$
and $p^{\mu}_{\;I}=\Lambda^{\mu}_{c} \uppi^{c}_{\;I}$
we state:

\begin{theorem*}
The DWH equations for the pair $(y^{I},p^{\mu}_{\;I})$ of field variables and polymomenta  assume the form:
        \begin{equation}
       \underline{\mathbf{\beta}}{}^{\mu}\partial_\mu \Psi_{(p)}= \nabla_{\!\!\mathcal{Z}_{(p)}}H. \label{q1g}
        \end{equation}
\end{theorem*}
\begin{proof}
First we compute the product on the l.h.s. of (\ref{q1g}),
\begin{align}
\underline{\mathbf{\beta}}{}^{\mu}\partial_\mu \Psi_{(p)} &=\Lambda^{\mu}_{\;c}\{({}^{c}\!P) -
\delta^{cd}(P_{\!d})\}\{\partial_{\mu}\textsf{y}^{I}(P_{I})+\partial_{\mu}\uppi^{\;I}_{a}({}^{a}\!P_{I})\} \nonumber \\
&=\Lambda^{\mu}_{\;c}\partial_{\mu}\textsf{y}^{I}({}^{c}\!P_{I})-\Lambda^{\mu}_{\;c}\delta^{cd} \delta^{a}_{d} \partial_{\mu}\uppi^{\;I}_{a}(P_{I})
\nonumber \\
&=\Lambda^{\mu}_{\;c} \partial_{\mu}\textsf{y}^{I} ({}^{c}\!P_{I})-\partial_{\mu}\textsf{p}^{\mu I}(P_{I}). \label{lhsg}
\end{align}
 Next we equate this result to the r.h.s.
\begin{equation}
\nabla_{\!\!\mathcal{Z}_{(p)}}H= (P_{I})\frac{\partial H}{\partial \textsf{y}_{I}}+
          ({}^{c}\!P_{I})\frac{\partial H}{\partial \uppi^{c}_{\;I}}\label{rhsg}
\end{equation}
of (\ref{q1g}) to obtain
\begin{equation}
\partial_{\mu}\textsf{p}_{\;I}^{\mu}=-\frac{\partial H}{\partial \textsf{y}^{I}}
\quad\text{and}\quad
\partial_{\mu}\textsf{y}^{I}=\frac{\partial H}{\partial \textsf{p}^{\mu}_{\;I}}. \nonumber
\end{equation}
\end{proof}
\section{Poisson-like Bracket}
In this section, we turn our attention to the bracket operation that, as will be shown shortly, can be considered as the analogue of the Poisson bracket for the DHW theory in the DKP representation.

 Let $\underline{\mathbf{\beta}}{}_{\mu}=(\Lambda^{-1}\!)^{\;a}_{\mu}\,\underline{\mathbf{b}}{}_{a}$ where $\underline{\mathbf{b}}{}_{a}$ is given in the definition   (\ref{dkds}) and
$$(\Lambda^{-1}\!)^{\;a}_{\mu}\Lambda^{\mu}_{\;b}=\delta^a_b.$$ We define the following formula  for the bracket:
\begin{definition}
\begin{equation}
\{G,F\}_{\mu}^{(p)}(P):=(p\,!)^{-1}\mathbf{C}_{(p)}[\{\nabla_{\!\!\mathcal{Z}_{(p)}}^{\dag}\!G\}
\underline{\mathbf{\beta}}{}_{\mu}\{\nabla_{\!\!\mathcal{Z}_{(p)}}F\}], \label{fb1}
\end{equation}
where the superscript $``(p)"$ continues to denote the rank of the antisymmetric fields $y^{I}$. $\mathbf{C}_{(p)}$ denotes the contraction operation defined in the Eq. (\ref{Contr}) of the Appendix  and
\begin{equation}
\nabla_{\!\!\mathcal{Z}_{(p)}}^{\dag}=({}^{I}\!P)\frac{\partial}{\partial \textsf{y}^I}+({}^{I}\!P_{\!a})\frac{\partial}{\partial \uppi^{\;I}_{a}}
\end{equation}
is the adjoint of the operator (\ref{dzop}), computed according to the formulas (\ref{adop1}) and (\ref{adop2}) (see Appendix).
\end{definition}
 We will apply the formula (\ref{fb1}) to compute explicitly the case $p=1$ in ample detail and $p=2$ with emphasis on the main steps. The case $p=0$ is the simplest one.

 So, for $p=1$,
\begin{align*}
&\{G,F\}_{\mu}^{(1)}(P)=\mathbf{C}_{(1)}[\{\nabla_{\!\!\mathcal{Z}_{(1)}}^{\dag}\!G\}
\underline{\mathbf{\beta}}{}_{\mu}\{\nabla_{\!\!\mathcal{Z}_{(1)}}F\}] \nonumber \\
&= \mathbf{C}_{(1)}\Biggl[\biggl\{({}^{j}\!P)\frac{\partial G }{\partial \textsf{y}^{j}}+
          ({}^{\;j}\!P_{i})\frac{\partial G }{\partial \uppi^{\;j}_{i}}\biggr\}\biggl\{(\Lambda^{-1})^{\;c}_{\mu}[(P_{\!c}) -
          \delta_{cd}({}^{d}\!P)]\biggr\}   \\
          &  \hspace*{30ex} \times \left\{(P_{\!k})\frac{\partial F }{\partial \textsf{y}_{k}}+
          ({}^{k}\!P_{l})\frac{\partial F}{\partial \uppi^{k}_{\;l}}\right\}\Biggl] \\
          &=\mathbf{C}_{(1)}\Biggl[\biggl\{({}^{j}\!P)\frac{\partial G}{\partial \textsf{y}^{j}}+
          ({}^{j}\!P_{\;i})\frac{\partial G}{\partial \uppi^{\;j}_{i}}\biggr\}\biggl\{(\Lambda^{-1})^{\;c}_{\mu}\delta^k_c(P_{l})\frac{\partial
          F}{\partial \uppi^{k}_{\;l}}   \\
          & \hspace*{30ex} -(\Lambda^{-1})^{\;c}_{\mu}\delta_{cd}({}^{d}\!P_{k})\frac{\partial F}{\partial
          \textsf{y}_k}\biggr\} \Biggr]  \\
          &=\mathbf{C}_{(1)}\left[({}^{j}\!P_{l})\frac{\partial G }{\partial \textsf{y}^{j}}\frac{\partial F}{\partial
          \textsf{p}^{\mu}_{\;l}}-({}^{j}\!P_{k})\frac{\partial G}{\partial \textsf{p}^{\,\mu j}}\frac{\partial F}{\partial
          \textsf{y}_{k}}\right] \\
&(l\rightarrow k \;\text{in the first term})= \left(\frac{\partial G }{\partial \textsf{y}^{j}}\frac{\partial F}{\partial \textsf{p}^{\mu}_{\;k}}-\frac{\partial
G}{\partial \textsf{p}^{\mu j}}\frac{\partial F}{\partial \textsf{y}_{k}}\right)\mathbf{C}_{(1)}\left[({}^{j}\!P_{k})\right] \\
&=\left(\frac{\partial G }{\partial \textsf{y}^{j}}\frac{\partial F}{\partial \textsf{p}^{\mu}_{\;k}}-\frac{\partial
G}{\partial \textsf{p}^{\mu j}}\frac{\partial F}{\partial \textsf{y}_{k}}\right)\delta^j_k(P) \nonumber \\
  &= \left(\frac{\partial G }{\partial \textsf{y}^{j}}\frac{\partial F}{\partial \textsf{p}^{\mu}_{\;j}}-\frac{\partial
  F}{\partial \textsf{y}^{j}}\frac{\partial G}{\partial \textsf{p}^{\mu}_{\;j}}\right)(P).
  \end{align*}
This bracket is similar to the one introduced earlier by Good \cite{Good} and Tapia \cite{Tapia} (see also \cite{MRO,JA}). In those references the Roman field indices are general.
  Here they describe  Euclidean antisymmetric fields -- see next case.

   We compute  the case $p=2$. The conjugate field variables are $(\textsf{y}^{ab}, \textsf{p}^{\mu}_{\;\,ab})$ where $\textsf{y}^{ab}=-\textsf{y}^{ba}$
   and $\textsf{p}^{\mu}_{\;\,ab}=-\textsf{p}^{\mu}_{\;\,ba}$. We have
\begin{align}
\{G,F\}_{\mu}^{(2)}(P)&=\frac{1}{2}\mathbf{C}_{(2)}[\{\nabla_{\!\!\mathcal{Z}_{(2)}}^{\dag}G\}
\underline{\mathbf{\beta}}{}_{\mu}\{\nabla_{\!\!\mathcal{Z}_{(2)}}F\}] \nonumber \\
%&=\cdot\cdot\cdot \nonumber \\
&=\frac{1}{2}\left(\frac{\partial G }{\partial \textsf{y}^{kt}}\frac{\partial F}{\partial \textsf{p}^{\mu}_{\;ab}}-\frac{\partial
G}{\partial \textsf{p}^{\mu kt}}\frac{\partial F}{\partial \textsf{y}_{ab}}\right)\mathbf{C}_{(2)}\left[({}^{kt}\!P_{ab})\right]
\nonumber \\
&=\frac{1}{2}\left(\frac{\partial G }{\partial \textsf{y}^{kt}}\frac{\partial F}{\partial \textsf{p}^{\mu}_{\;ab}}-\frac{\partial
G}{\partial \textsf{p}^{\mu kt}}\frac{\partial F}{\partial y_{ab}}\right)(\delta^k_b\delta^t_a-\delta^t_b\delta^k_a)(P) \nonumber \\
  &= \left(\frac{\partial G }{\partial \textsf{y}^{ab}}\frac{\partial F}{\partial \textsf{p}^{\mu}_{\,\;ab}}-\frac{\partial
  G}{\partial \textsf{p}^{\mu}_{\,\;ab}}\frac{\partial F}{\partial \textsf{y}^{ab}}\right)\!(P).\nonumber
 \end{align}
 We can proceed to higher values of $p$ up to $p=n=\dim M$. We conclude that the general expression for the bracket for antisymmetric fields of rank
 $0\leq p \leq n$ is as follows:
 \begin{equation}
 \{G,F\}_{\mu}^{(p)}(P)=\left(\frac{\partial G }{\partial \textsf{y}^{I}}\frac{\partial F}{\partial
 \textsf{p}^{\mu}_{\,\;I}}-\frac{\partial
  G}{\partial \textsf{p}^{\mu}_{\,\;I}}\frac{\partial F}{\partial \textsf{y}^{I}}\right)\!(P), \nonumber
\end{equation}
where $I$ is the multi-index of length ``$p$".

It is easily  seen that the bracket (\ref{fb1}) is antisymmetric and fulfills the Leibniz rule
\begin{align*}
&\{GF,K\}_{\mu}^{(p)}(P)=(p\,!)^{-1}\mathbf{C}_{(p)}[\{\nabla_{\!\!\mathcal{Z}_{(p)}}^{\dag}\!GF\}
\underline{\mathbf{\beta}}{}_{\mu}\{\nabla_{\!\!\mathcal{Z}_{(p)}}K\}] \\
&=\left\{(p\,!)^{-1}\mathbf{C}_{(p)}[\{\nabla_{\!\!\mathcal{Z}_{(p)}}^{\dag}\!G\}
\underline{\mathbf{\beta}}{}_{\mu}\{\nabla_{\!\!\mathcal{Z}_{(p)}}K\}]\right\}F  \\
 & \hspace*{30ex} +G\left\{(p!)^{-1}\mathbf{C}_{(p)}[\{\nabla_{\!\!\mathcal{Z}_{(p)}}^{\dag}\!F\}
\underline{\mathbf{\beta}}{}_{\mu}\{\nabla_{\!\!\mathcal{Z}_{(p)}}K\}]\right\} \\
&=F\{G,K\}_{\mu}^{(p)}(P)+G\{F,K\}_{\mu}^{(p)}(P).
\end{align*}

In order to verify the analogue of the Jacobi identity we use the formula (\ref{fb1})  to compute
\begin{align}
&\{\{G,F\}_{\mu}^{(p)},\mathrm{K} \}^{(p)}_{\nu}(P)=\!(p\,!)^{-1}\mathbf{C}_{(p)}\!\biggl[\left\{\nabla_{\!\!\mathcal{Z}_{(p)}}^{\dag}   \!\!  \left(\{G,F\}_{\mu}^{(p)} \right)\! \right\}
\underline{\mathbf{\beta}}{}_{\nu}\{\nabla_{\!\!\mathcal{Z}_{(p)}}K\}\biggr]\nonumber \\
&=\!(p\,!)^{-1}\mathbf{C}_{(p)}\!\Biggl[\biggl\{\nabla_{\!\!\mathcal{Z}_{(p)}}^{\dag}\!\!  \left(\frac{\partial G }{\partial \textsf{y}^{I}}\frac{\partial F}{\partial
 \textsf{p}^{\mu}_{\,\;I}}-\frac{\partial
  G}{\partial \textsf{p}^{\mu}_{\,\;I}}\frac{\partial F}{\partial \textsf{y}^{I}}\!\biggl)\! \biggl\}
\underline{\mathbf{\beta}}{}_{\nu}\{\nabla_{\!\!\mathcal{Z}_{(p)}}K \} \right] \nonumber \\
&=\!(p\,!)^{-1}\mathbf{C}_{(p)}\!\Biggl[\left\{\nabla_{\!\!\mathcal{Z}_{(p)}}^{\dag}\!\!  \left(\frac{\partial G }{\partial \textsf{y}^{I}}\frac{\partial F}{\partial
\textsf{p}^{\mu}_{\,\;I}}-\frac{\partial G}{\partial \textsf{p}^{\mu}_{\,\;I}}\frac{\partial F}{\partial \textsf{y}^{I}}\!\right)\! \right\} \nonumber \\
& \hspace{3cm} \times \left\{(P_L)\frac{\partial K}{\partial \textsf{p}^{\nu}_{\,\;L}}-\delta_{cd}(\Lambda^{-1})^c_{\nu}({}^{d}\!P_{L})\frac{\partial K}{\partial \textsf{y}_L}\right\}\Biggr]. \label{IJB}
\end{align}
 The first term in the product of the two curly brackets in the above expression expands as follows:
\begin{align}
&\nabla_{\!\!\mathcal{Z}_{(p)}}^{\dag}\!\!  \left(\frac{\partial G }{\partial \textsf{y}^{I}}\frac{\partial F}{\partial
\textsf{p}^{\mu}_{\,\;I}}-\frac{\partial G}{\partial \textsf{p}^{\mu}_{\,\;I}}\frac{\partial F}{\partial \textsf{y}^{I}}\!\right)\! =\left(({}^{J}\!P)\frac{\partial}{\partial \textsf{y}^J}+({}^{I}\!P_{\!a})\frac{\partial}{\partial \uppi^{\;J}_{a}}\right) \nonumber \\
& \hspace{6cm}\times \left(\frac{\partial G }{\partial \textsf{y}^{I}}\frac{\partial F}{\partial
\textsf{p}^{\mu}_{\,\;I}}-\frac{\partial G}{\partial \textsf{p}^{\mu}_{\,\;I}}\frac{\partial F}{\partial \textsf{y}^{I}}\!\right)\nonumber \\
&=({}^{J}\!P)\left[\frac{\partial^2G}{\partial \textsf{y}^J\partial \textsf{y}^I}\frac{\partial F}{\partial
\textsf{p}^{\mu}_{\,\;I}}+\frac{\partial G}{\partial \textsf{y}^I}\frac{\partial^2F}{\partial
\textsf{y}^J \partial \textsf{p}^{\mu}_{\,\;I}}-\frac{\partial^2G}{\partial
\textsf{y}^J \partial \textsf{p}^{\mu}_{\,\;I}}\frac{\partial F}{\partial \textsf{y}^I}
-\frac{\partial G}{\partial\textsf{p}^{\mu}_{\,\;I}}\frac{\partial^2F}{\partial \textsf{y}^J\partial \textsf{y}^I}\right]\nonumber \\
&\hspace{-0,5cm}+({}^{J}\!P_{a})\left[\frac{\partial^2G}{\partial \uppi^{\;J}_{a}\partial \textsf{y}^I}\frac{\partial F}{\partial
\textsf{p}^{\mu}_{\,\;I}}+\frac{\partial G}{\partial
\textsf{y}^I}\frac{\partial^2F}{\partial \uppi^{\;J}_{a}\partial \textsf{p}^{\mu}_{\,\;I}}-\frac{\partial^2G}{\partial \uppi^{\;J}_{a}\partial \textsf{p}^{\mu}_{\,\;I}}\frac{\partial F}{\partial
\textsf{y}^I}-\frac{\partial G}{\partial
\textsf{p}^{\mu}_{\,\;I}}\frac{\partial^2F}{\partial \uppi^{\;J}_{a}\partial \textsf{y}^I}
\right].\nonumber
\end{align}
Multiplying this result from the left by the second curly bracket in (\ref{IJB}) and further computing the contraction $(p\,!)^{-1}\mathbf{C}_{(p)}$
 we obtain
 \begin{align}
& \{\{G,F\}^{(p)}_{\mu},\mathrm{K} \}^{(p)}_{\nu}(P)=\Biggl[\frac{\partial^2G}{\partial \textsf{y}^J\partial y^I}\frac{\partial F}{\partial
\textsf{p}^{\mu}_{\,\;I}}\frac{\partial K}{\partial \textsf{p}^{\nu}_{\,\;J}}+\frac{\partial G}{\partial \textsf{y}^I}\frac{\partial^2F}{\partial
\textsf{y}^J \partial \textsf{p}^{\mu}_{\,\;I}}\frac{\partial K}{\partial \textsf{p}^{\nu}_{\,\;J}} \nonumber \\
&-\frac{\partial^2G}{\partial
\textsf{y}^J \partial \textsf{p}^{\mu}_{\,\;I}}\frac{\partial F}{\partial \textsf{y}^I}\frac{\partial K}{\partial \textsf{p}^{\nu}_{\,\;J}}
-\frac{\partial G}{\partial\textsf{p}^{\mu}_{\,\;I}}\frac{\partial^2F}{\partial \textsf{y}^J\partial \textsf{y}^I}\frac{\partial K}{\partial \textsf{p}^{\nu}_{\,\;J}}-
\frac{\partial^2G}{\partial \textsf{p}^{\nu}_{\,\;J}\partial \textsf{y}^I}\frac{\partial F}{\partial
\textsf{p}^{\mu}_{\,\;I}}\frac{\partial K}{\partial \textsf{y}^J} \nonumber \\
&-\frac{\partial G}{\partial
\textsf{y}^I}\frac{\partial^2F}{\partial \textsf{p}^{\nu}_{\,\;J}\partial \textsf{p}^{\mu}_{\,\;I}}\frac{\partial K}{\partial \textsf{y}^J}+\frac{\partial^2G}{\partial \textsf{p}^{\nu}_{\,\;J}\partial \textsf{p}^{\mu}_{\,\;I}}\frac{\partial F}{\partial
\textsf{y}^I}\frac{\partial K}{\partial \textsf{y}^J}+\frac{\partial G}{\partial
\textsf{p}^{\mu}_{\,\;I}}\frac{\partial^2F}{\partial \textsf{p}^{\nu}_{\,\;J}\partial \textsf{y}^I}\frac{\partial K}{\partial \textsf{y}^J}\Biggr]
(P),\nonumber
\end{align}
which is the expected result. We would like to call the attention to the role played by the algebraic objects $(P_{L})$,  $({}^{d}\!P_{L})$, $({}^{J}\!P)$ and $({}^{J}\!P_{a})$  in these calculations. They occur in the situations:
\begin{gather*}
({}^{J}\!P)(P_{L})=({}^{J}\!P_{L}), \quad ({}^{J}\!P)({}^{d}\!P_{L})=0, \quad ({}^{J}\!P_{a})(P_{L})=0\\
\text{and} \quad ({}^{J}\!P_{a})({}^{d}\!P_{L})=\delta^d_a({}^{J}\!P_{L}).
\end{gather*}

Finally, with some more effort, one can check that the following generalization of the Jacobi identity is satisfied:
\begin{equation}
\{\{G,F\}_{(\mu},K \}^{(p)}_{\nu)}(P)+\mathrm{cyclic}(G,F,K)=0, \label{FJI}
\end{equation}
where $A_{(\mu\nu)}:=\frac{1}{2}(A_{\mu\nu}+A_{\nu\mu})$. Therefore the formula (\ref{fb1}) can be considered as an analogue of the Poisson bracket in the DKP matrix representation-free form for the DHW theory for antisymmetric fields. This result generalizes the result of ref. \cite{Igor1}. It is worth noting that formulas (\ref{fb1}) and (\ref{FJI}) have similar form to the equivalents obtained in \cite{Igor1}, which are in fact the analogues of the traditional matrix formulation of the Poisson bracket. Here, $\underline{\mathbf{\beta}}{}_{\mu}$ plays the role of the Poisson structure.
\section{Conclusion}
  The action of the DKP algebra on invariant subspaces of the Clifford algebra, including the spaces of antisymmetric fields, was used to formulate the DKP form of the covariant DWH field equations within an algebraic framework independent of matrix representation.

By means of an adequate choice of the sign in the construction of the DKP-generators the modified elements e.g., (\ref{nmdkp}) have led to the correct expression for the $\underline{\mathbf{\beta}}{}^{\mu}$-generators which satisfy the relation (\ref{ndkc})  needed for the DKP formulation of the DWH equations. Effectively, the DKP operator in the form (\ref{nmdkp}) and its action according to Proposition~2 explain why these generators represent  the  $k$-symplectic structure on the space of pairs of field variables $(\textsf{y}^{I},\textsf{p}_{\,\;I}^{\mu})$ of the polymomentum phase space. This leads us to the conclusion that the space of these pairs, as defined here, is the natural space of the conjugate variables for the DWH theory of antisymmetric fields. These results were summarized in the theorem of section 8 which establishes a simple matrix representation-free form for the DWH field equations using the DKP first order differential operator.

  Finally, we have found the formula for a bracket operation that shows consistency with the DWH theory and fulfills a generalization of the properties of the Poisson bracket. In this formula, the  $\underline{\mathbf{\beta}}{}_{\mu}$-generators of the DKP algebra play the  role analogous to that of a  Poisson structure, performing exactly as one would expect when compared to the traditional matrix formulation of the Poisson bracket in mechanics. The algebraic operations resulting from the formula  contributed decisively to the results achieved. Such operations originated from a fully assembled multilinear calculus based on the relationship between the Clifford algebra and the DKP algebra as it was pioneered by Sch\"{o}nberg \cite{Schom}.

  The vectorial nature of this type of bracket is related to the arbitrariness in the choice of the direction of time, as suggested in \cite{Marsden}. This is important for the 3+1 decomposition in the traditional description of the relativistic dynamics of fields. It gives rise to the lapse and shift. Our approach to the bracket brings new perspectives for the advancement of research on this subject and we hope to return to this matter in another publication.

  Overall, we have analyzed many useful aspects of the relationship between the Clifford and the DKP algebras. We believe this to be of practical interest for future applications, in particular the formulation of both algebras in terms of the projector basis (multilinear endomorphisms). This strategy allows direct access to the decomposition of these algebras and has proved to be very useful for the study of their representations.

  It is now interesting to investigate how to extend the DKP formulation of the DWH theory to symmetric tensor fields, with the aim of including metric fields and, hopefully, gravity. It is our intention to address this issue in a future publication.
\section*{Acknowledgments}
I would like to express my gratitude to J. D. M. Vianna for bringing Ref.\cite{Schom} to my attention in the early stages of my academic career and for the countless discussions on the fascinating insights into the geometric algebras contained in that reference. I also thank Jos\'{e} F. R. Neto for reading early drafts of the manuscript and making useful remarks. Finally, I would also like to thank the referees for the comments and suggestions which greatly contributed to the improvement of this work. Special thanks are due to the editor for making the article better.

%\appendix
\section*{Appendix: Multilinear Endomorphisms in $\mathrm{G}_n$}

  In this Appendix we show that the splitting $W=V\oplus V^*$ of the vector space $W$ into two complementary Lagrangian
  subspaces $V$ and $V^*$ under the bilinear form, relation (\ref{bf1}), in addition with the corresponding invariant projector $(P)$, lead to a
   useful basis to expand the algebra $\mathrm{G}_n$. The convenience of this type of basis is that it makes explicit the appearance of the
  spaces of algebraic spinors  \cite{Schom,M.Fernandes,Cartan,Bergdolt1,Bergdolt2} in $\mathrm{G}_n$. These spaces turn out to be the minimal left and
  right ideals of the algebra. We will also introduce the multi-index notation which is particularly useful when dealing with multilinear algebras.

\subsection*{The Projector Basis of the Split Form}

Start with a dual basis $e_1,\ldots, e_n$ and $e^1,\ldots, e^n$ of vectors and covectors respectively of the complementary Lagrangian subspaces $V$ and $V^*$. Set the invariant projectors as follows,
\begin{equation}
\Pi_{p}=(p\,!)^{-1}\sum_{j_1=1}^{n}\cdots\sum_{j_p=1}^{n} {}^{j_1,\ldots,j_p}\!P_{j_p,\ldots,j_1}, \label{pro}
\end{equation}
where
\begin{equation}
{}^{j_1,\ldots,j_p}\!P_{j_p,\ldots,j_1}:= e^{j_1}\cdots e^{j_p}(P)e_{j_p}\cdots e_{j_1},\label{gips}
\end{equation}
and $p=0,\ldots,n$, $\dim V =n$, with the convention that $\Pi_{0}:=P= {\mathcal{N}}_{1}\cdots{\mathcal{N}}_{n}=e_{1}e^{1}\cdots e_{n}e^{n}$. Notice that
the projectors (\ref{pro}) are not postulated but they are constructed from Clifford products of isotropic basis vectors. This construction is due to Sch\"{o}nberg \cite{Schom}.

A multi-index $I$ of length $k$ is a $k$-tuple of positive integers, say $i_1 i_2\ldots i_k$, from the set $\{1,2,\ldots,\dim V\}$. We are going to denote them by the upper-case Roman letters and the associated indices will be denoted by the lower case letters. A multi-index $I$ is said to have the length $p$ if $k=p$ in the $k$-tuple that represents $I$ and in that case, we write $|I|=p$ for the length of $I$. We will also use the summation convention for the multi-index.

Using the multi-index notation we rewrite equations (\ref{pro}) and (\ref{gips}) accordingly,
\begin{equation}
\Pi_{|J|}=(|J|!)^{-1}\sum_{J} {}^{J}\!P_{\overleftarrow{J}}, \label{promi}
\end{equation}
where $J=j_{1}\ldots j_{|J|}$; $\overleftarrow{J}:=j_{|J|}\ldots j_1$ and $|\overleftarrow{J}|=|J|=p$.  The length $|K|$ of the elements can only runs from $0$ to dim $V$ in the $\mathrm{G}_n$ algebra.

As a consequence of the relations (\ref{gn1}), (\ref{gn2}), (\ref{gn3}) and (\ref{fr1}), the elements (\ref{gips}) satisfy the relation,
\begin{equation}
({}^{j_1,\ldots,j_p}\!P_{k_q,\ldots,k_1})({}^{j'_{1},\ldots,j'_{p'}}\!P_{k'_{q'},\ldots,k'_{1}})=
\delta_{q,p'}\delta^{j'_{1},\ldots,j'_{p'}}_{k_1,\ldots,k_q}({}^{j_1,\ldots,j_p}\!P_{k'_{q'},\ldots,k'_{1}}) \label{gg},
\end{equation}
where $\delta_{q,p'}\delta^{j'_{1},\ldots,j'_{p'}}_{k_1,\ldots,k_q}$ denotes the generalized Kronecker deltas,
\begin{gather*}
\delta_{q,p'} \delta^{j'_{1},\ldots,j'_{p'}}_{k_1,\ldots,k_q} = \det\left(
\begin{array}{cccc}
          \delta^{j'_1}_{k_1}& \delta^{j'_1}_{k_2}& \cdots & \delta^{j'_{1'}}_{k_q} \\
          \delta^{j'_2}_{k_1} & \delta^{j'_2}_{k_2} & \cdots & \delta^{j'_{2'}}_{k_q} \\
           \cdot & \cdot & \cdot & \cdot \\
            \delta^{j'_{p'}}_{k_1} & \delta^{j'_{p'}}_{k_2} & \cdots & \delta^{j'_{p'}}_{k_q} % \\
\end{array}
\right),
\end{gather*}
or succinctly stated,
\begin{equation}
({}^{J}\!P_{\overleftarrow{K}})({}^{J'}\!P_{\overleftarrow{K'}})=\delta_{|\overleftarrow{K}|,|J'|}  \Delta^{J'}_{K} ({}^{J}\!P_{\overleftarrow{K'}}), \qquad
|\overleftarrow{K}|=q, \; |J'|=p', \label{cis}
\end{equation}
where $\Delta^{J'}_{K}\equiv\delta^{j'_{1},\ldots,j'_{p'}}_{k_1,\ldots,k_q}$.

Relation (\ref{gg}) or (\ref{cis}) shows that the set of $2^{2n}$ elements (\ref{gips}) of $\mathrm{G}_n$ with $p,q=0,\ldots,n$ and $j_{1}<\cdots<j_{p}$, $k_{1}<\cdots<k_{q}$, $j, k=1,\ldots,n$ are linearly independent and thus form a basis of $\mathrm{G}_n$, which has dimension $2^{2n}$.

In particular, the generators $(e_j)$ and $(e^j)$ can be retrieved from the basis elements ${}^{j_1,\ldots,j_p}\!P_{k_q,\ldots,k_1}$ by writing
\begin{align*}
e^j&=\sum_p (p!)^{-1}({}^{j,j_1,\ldots,j_p}\!P_{ j_p,\ldots,j_1})=\sum_{|J|=0}^{n} (|J|!)^{-1}\;({}^{j,J}\!P_{\overleftarrow{J}}),  \\
e_j&=\sum_p (p!)^{-1}({}^{j_1,\ldots,j_p}\!P_{j_p,\ldots,j_1,j})=\sum_{|J|=0}^{n} (|J|!)^{-1}({}^{J}\!P_{\overleftarrow{J},j}),
\end{align*}
which are seen to satisfy the defining relations (\ref{gn1})--(\ref{gn3}) of the $\mathrm{G}_n$ algebra according to multiplication (\ref{gg}) or (\ref{cis}). Any element of the $\mathrm{G}_n$ algebra can be written in the basis $(^{J} P_{\overleftarrow{K}})$ as follows:
\begin{align}
\Lambda&=\sum_{p,q=0}^n (p!q!)^{-1}  A_{j_1,\ldots,j_p}{}^{k_1,\ldots,k_q}  ({}^{j_1,\ldots,j_p}\!P_{k_q,\ldots,k_1}) \nonumber \\
&=\sum_{|J|,|K|=0}^{n} A_{J}{}^{K}({}^{J}\! P_{\overleftarrow{K}}), \label{egn}
\end{align}
where we can recognize the coefficients $A_{j_1,\ldots,j_p}{}^{k_1,\ldots,k_q}$ as tensor components. Note that, by writing all the terms of the sum explicitly there will occur all types of tensors including the zeroth rank elements regarded as the scalars, the mixed tensors and finally the antisymmetric ones, either covariant or contravariant up to rank $n$. The great advantage of studying the $\mathrm{G}_n$ algebra in this basis is that by using the system of multilinear projectors $\Pi_{|J|}$ we can easily project on tensor spaces of different lengths within $\mathrm{G}_n$. One of the most important projectors is  the minimal idempotent $\Pi_0=(P)$. It projects $\mathrm{G}_n$ onto its minimal
left/right ideals. That is, take $\Lambda \in \mathrm{G}_n$ written in the form (\ref{egn}) and compute the product or right projection
$\Lambda(\Pi_0)=\Lambda(P)$ using rules (\ref{gg}). The result is
\begin{equation}
\psi=\sum_{p=0}^n (p!)^{-1} A_{j_1,\ldots,j_p} ({}^{j_1,\ldots,j_p}\!P)=\sum_{|J|=0}^n (|J|!)^{-1}A_J({}^J \! P). \nonumber
\end{equation}
Let us call the space of these elements  $\mathrm{G}_n(P)$. Such space is isomorphic to the direct sum of the linear spaces of homogeneous forms of all
orders up to $n=\dim V$. The set $\mathrm{G}_n(P)$  is a minimal left ideal of $\mathrm{G}_n$. Clearly, due to rules (\ref{gg}) one can easily check that the
elements $\Lambda \in \mathrm{G}_n$ are endomorphisms of $\mathrm{G}_n(P)$. So $\mathrm{G}_n(P)$ is the space of spinors of the vector space $W=V\oplus V^{*}$
with which we started.

Upon left projection, the analogous arguments lead to the minimal right ideal $(P)\mathrm{G}_n$ which is the dual of $\mathrm{G}_n(P)$.

 It is common in the literature  to define these spinors  as elements of the Grassmann algebra of the Lagrangian subspace $V^{*}$ of $W$ \cite{Mein,
 Deline}. This Grassmann algebra turns out to be a Clifford module.  The dual spinors belong to the Grassmann algebra of the complementary space $V$.
 These Clifford modules are isomorphic to the minimal left/right ideals of $\mathrm{G}_n$. These minimal ideals are known as algebraic spinor spaces
 \cite{Schom, Riesz, Cartan, Bergdolt1, Bergdolt2}.

By using the system of multilinear projectors $\Pi_p$, we see that the projection $\Pi_p\mathrm{G}_n(P)$ is isomorphic to the space of homogeneous forms
of order (rank) $p$. By using the $(\Pi_p)$ we can also write an element of $\mathrm{G}_n(P)$ in a basis free form

\begin{equation}
\psi=\sum_{|J|=0}^n (\Pi_{|J|})\Lambda(P), \nonumber
\end{equation}
where $\Lambda \in \text{G}_n$.

From the multiplication rule (\ref{gg}) the following useful algebraic properties can be obtained which are used throughout the text
          \begin{equation}
          \sum_{|J|=0}^{n}  \Pi_{|J|}=1_{G_n}, \quad   \Pi_{|J|}(\Pi_{|K|})=\delta_{|J|,|K|}(\Pi_{|J|}) \nonumber
          \end{equation}
          and
              \begin{equation}
              \alpha(\Pi_{|J|})=(\Pi_{|J|+1})\alpha, \qquad   \Pi_{|J|}(v)=v(\Pi_{|J|+1}). \label{alp2}
              \end{equation}
In plain words, $\Pi_{|J|}$ changes to $\Pi_{|J|+1}$ when it moves left past a covariant vector, and $\Pi_{|J|}$ changes to $\Pi_{|J|+1}$ when it
moves right past a contravariant vector. We also assume that $\Pi_{-1}=0$ and $\Pi_{n+1}=0$.

\subsection{Adjunction and Contraction}
 Two further operations in the algebra can be defined. The first operation is \textbf{adjunction} which is defined  as follows \cite{Schom},
\begin{equation}
(e_i)^{\dag}=(e_i,0)^{\dag}:=(0,g_{ij}e^j)=(0,\widetilde{e}_i)=g_{ij}(e^j)\equiv (\widetilde{e}_i)\label{adop1}
\end{equation}
and similarly,
\begin{equation}
(e^i)^{\dag}=(0,e^i)^{\dag}:=(g^{ij}e_j,0)= (\tend{e}^i,0)=g^{ij}(e_j)\equiv(\tend{e}^i) \label{adop2}
\end{equation}
and hence $(P)^{\dag}=(P)$. The adjunction is an involution of $\mathrm{G}_n$. In the case of $\mathrm{G}_n$ over the real numbers the adjunction becomes transposition which is an involution corresponding to the transformation $v\rightarrow \sum v^i e^i$, $\alpha\rightarrow \sum \alpha_i e_i$. In the complex case the involution
corresponds to the transformation $v\rightarrow \sum (v^i)^*e_i^{\dag}$ associated to a unitary metric $g_{ij}=h_{ij}$ \cite{Schom}.

The other operation is \textbf{contraction}. Recall that the standard operation of contraction defined for tensors fields shrinks an $(r,s)$ tensor to an
$(r-1,s-1)$ tensor. For example, for $(1,1)$  tensor fields on a manifold $M$, the contraction shrinks into functions.  It is defined as the
$\mathfrak{F}(M)$-linear map
\begin{equation}
\mathbf{C}: \mathfrak{T}^1_1\rightarrow \mathfrak{F}(M), \qquad \mathbf{C}(X\otimes \theta)=\langle \theta, X\rangle \nonumber,
\end{equation}
for all one-forms $\theta$ and vector fields $X$. In the local chart $(U,x^1\cdots x^n)$ we must have $\mathbf{C}(dx^i\otimes\partial_j)=\delta^i_j$. So for a tensor field $A \in \mathfrak{T}^1_1$, we obtain
\begin{equation}
\mathbf{C}\left(\sum A_{\;\;i}^j\partial_j \otimes dx^i \right)=\sum A^i_{\;\;i}=A_{\;\;i}^i. \nonumber
\end{equation}
Einstein summation rule has been employed.

In analogy with this operation on tensors, we define a similar operation of \textbf{contraction in the algebra} $\mathrm{G}_n$. For example, for the elements
$A_i^j ({}^{i}\!P_{\!j})$ we define:
\begin{align}
\mathbf{C}[A_i^{\;\;j} ({}^{i}\!P_{\!j})]&=A_i^{\;\;j}\mathbf{C}[{}^{i}\!P_{\!j}]:=A_i^{\;\;j}(P_{\!j})({}^{i}\!P) \nonumber \\
&= A_i^{\;\;j} \delta^i_{j}(P)=A^{\;\;i}_{i}(P)\label{contr}
\end{align}
where the property (\ref{d1a}) and the idempotency of $(P)$ have been used. We extend the contraction operation to the basis elements of $\mathrm{G}_n$  as follows:
\begin{equation}
\mathbf{C}_{(p)}\left[{}^{J}\!P_{\overleftarrow{K}}\right]:= \Delta^{J}_{K}(P) \label{Contr}
\end{equation}
with $|J|=|K|=p$.
This full contraction shrinks the elements ${}^{J}\!P_{\overleftarrow{K}}$  to products $s(P)$ of the projectors $P$ in the subspace of the scalars in
$\mathrm{G}_n$ with~$s$ an element of~$\mathbb{F}$.

\end{document}